\documentclass[twoside,reqno,10pt]{amsart}

\usepackage{amsmath, amssymb, amsfonts}
\usepackage{url}
\usepackage{graphicx, color}

 \newtheorem{theorem}{Theorem}
 \newtheorem{lemma}{Lemma}
 \newtheorem{corollary}{Corollary}
 \newtheorem{proposition}{Proposition}

 \newtheorem{definition}{Definition}
 \newtheorem{remark}{Remark}

 \newcommand{\fclass}[2]{\ensuremath{  \mathbb{#1}^{\, #2} }}
 
 \usepackage{amsaddr}
 
\begin{document} 
 

 \title[Analytical Solution of the SIR model]{Analytical parameter estimation of the SIR epidemic model. Applications to the COVID-19 pandemic. }
 	
 	\author{Dimiter Prodanov }
\address{ 1) Environment, Health and Safety, IMEC, Belgium \\
	 2) MMSIP, IICT, BAS, Bulgaria}
%


 
 \begin{abstract}
 	The dramatic outbreak of the coronavirus disease 2019 (COVID-19) pandemics and its ongoing progression boosted the scientific community's interest in epidemic modeling and forecasting. 
 	The SIR (Susceptible-Infected-Removed) model is a simple mathematical model of epidemic outbreaks, yet for decades it evaded the efforts of the community to derive an explicit solution. 
 	The present work demonstrates that this is a non-trivial task.  
 	Notably, it is proven that the explicit solution of the model requires the introduction of a new transcendental special function, related to the Wright's Omega function. 
 	The present manuscript reports new analytical results and numerical routines suitable for parametric estimation of the SIR model.  	
 	The manuscript introduces iterative algorithms approximating the incidence variable, which allows for estimation of the model parameters from the numbers of observed cases. 
 	The numerical approach is exemplified with data from the European Centre for Disease Prevention and Control (ECDC) for several European countries in the period Jan 2020 -- Jun 2020. 
 \end{abstract}
 \maketitle
	
 Keywords: SIR model; special functions; Lambert W function; Wright Omega function
MSC: 92D30; 92C60; 26A36; 33F05; 65L09

 \section{Introduction}\label{sec:intro}

The coronavirus 2019 (COVID-19) disease was reported to appear for the first time in Wuhan, China, and later it spread to Europe, which is the subject of the presented case studies,  and eventually worldwide.
While there are individual clinical reports for COVID-19 re-infections, the present stage of the pandemic still allows for the application of a relatively simple epidemic model, which is the subject of the present report. 
The motivation behind the presented research was the intention to accurately model the short-term dynamics of the outbreaks of COVID-19 pandemics, which by the time of writing, has infected more than 27 million individuals worldwide. 
The efforts to contain the spread of the pandemic induce sustained social and economic damage. 
Therefore, the ability to accurately forecast short to medium-term  epidemic  outbreak's dynamics is of substantial public interest.  

The present manuscript gives a comprehensive analytical and numerical treatment of the SIR (Susceptible-Infected-Removed) epidemiological model. 
The SIR model was introduced in 1927 by Kermack and McKendrick in 1927 to study the plague and cholera epidemics in London and Bombay \cite{Kermack1927}.
To date the SIR model remains as a cornerstone of mathematical epidemiology.  
It is a deterministic model formulated in terms of ordinary differential equations (ODEs). 
The model has been extensively used to study the spread of various infectious diseases (see the monograph of Martcheva \cite{Martcheva2015}).

The objective of the present paper is to demonstrate numerical routines for curve-fitting allowing for estimation of the parameters of the SIR model from empirical data. 
To this end, the paper exhibits an algorithm, which can be used to compute the population variables as functions of time. 
In contrast to previous approaches, I do not consider the SIR model as an initial value problem but as a problem in the theory of  special functions.
This change of perspective allows for handling noisy data, e.g. time series having fluctuations caused by delays and accumulation of case reporting. 
The numerical approach is applied to the COVID-19 incidence   and case fatality data in  different European countries, having different population densities and dynamics of the epidemic outbreaks.

 \section{The SIR model}\label{sec:sir}
 
 The SIR model is formulated in terms of 3 populations of individuals.
 The \textit{S} population consists of all individuals susceptible to the infection of concern. 
 The \textit{I} population comprises the infected individuals.  
 These persons have the disease and can transmit it to the susceptible individuals.
 The \textit{R} poulation  cannot become infected and the individuals cannot transmit the disease to others.
 The model comprises a set of three ODEs:
 \begin{align}
 \dot{S }(t) & = - \frac{\beta}{N} S(t) I(t) \\
 \dot{I }(t) & =  \frac{\beta}{N} S(t) I(t) - \gamma I(t) \\
 \dot{R }(t) & =  \gamma I(t)
 \end{align}
 The model assumes a constant overall population $N= S + I + R$.
  An disease carrier infects on average $\beta$ individuals per day, for an average time of $1/\gamma$ days. The $\beta$ parameter is called \textit{disease transmission rate}, while $\gamma$ -- \textit{recovery rate}.
  The average number of infections arising from an infected individual is then modelled by the number $R_0=\frac{\beta}{\gamma}$, the \textit{basic reproduction number}.
  Typical initial conditions are 
  $ S(0)=S_0, I(0)=I_0, R(0)=0  $ \cite{Kermack1927}.
  
 The model can be re-parametrized using normalized variables as
 \begin{align}
 \dot{s} & = - s i \label{eq:ds1} \\
 \dot{i} & =   s i  - g i, \quad g= \frac{\gamma}{\beta} =\frac{1}{R_0} \label{eq:di1}\\
 \dot{r} & =   g i,
 \end{align}
 subject to normalization $s+i+r=1$ and time rescaling $ \tau =\beta t$.
 Therefore, since $i(\tau)$ is integrable on $[0, \infty )$ then $i(\infty)=0$.
 \section{The analytical solution}
The analytical solution will be formulated first in an implicit form.
Since there is a first integral by construction the system can be reduced to two equations:
  \begin{align}
 \frac{di}{ds} &= -1 + \frac{g}{s} \label{eq:dids} \\
 \frac{di}{dr} & = \frac{s}{g} -1 \label{eq:didr}
 \end{align}
\begin{remark}
	From this formulation
	 \[
	R_e=N\frac{s_0}{g}=\frac{S_0 \beta}{\gamma} \geq 1
	\]
	must hold for the infection to propagate. 
	$R_e$ is called the effective reproductive number, while the basic reproduction number is $R_0=R_e N$ \cite{Weiss2013}.
\end{remark}

In order to solve the model we will consider the two equations separately. 
 Direct integration of the   equation \ref{eq:dids} gives
 \[
 i = -s + g \log{s} +c
 \]
 where the constant $c$ can be determined from the initial conditions.
 In the present treatment, the constant $c$ will be left indeterminate to be assigned by the different re-parametrization procedures. 
 The \textit{s} variable can be expressed explicitly in terms of the Lambert W function \cite{Corless1996}:
 \begin{equation}
  s= - g W_{\pm} \left(-  \frac{e^\frac{i-c}{g}}{g} \right) 
 \end{equation}
where the signs denote the two different real-valued branches of the function. 
Note, that both branches are of interest since the argument of the Lambert W function is negative.  
 Therefore, the ODE \ref{eq:di1} can be reduced to the first-order autonomous system  
 \begin{equation}\label{eq:isr}
\dot{i}=- i  g \left( W_{\pm} \left(- \frac{ e^\frac{i-c}{g}}{g} \right) +1\right) 
 \end{equation}
valid for two disjoined domains on the real line. 
The ODEs can be solved for the time $\tau$ as
 \begin{equation}\label{eq:invi}
 -\int \frac{ d i}{i   \left( W_{\pm} \left(-  \frac{e^\frac{i-c}{g}}{g} \right) +1 \right) } = g  \tau
 \end{equation}
\begin{remark}\label{rem:omega}
	There is another equivalent form of the system using the Wright $\Omega$ function \cite{Corless2002}	since
	\[
	{W}\left( -\frac{{{ e}^{\frac{i-c}{g}
	}}}{g}\right) = \Omega \left( \frac{i-c}{g} - \log{(-g)} \right) 
	\]
	so that 
	\[
	\dot{i}=- i  g \left( \Omega_{\pm} \left(\frac{i-c - g\log{(-g)}}{g}  \right) +1\right)
	\]
\end{remark}

 The \textit{s} variable can be determined by substitution in  equation  \ref{eq:ds1}, resulting in the autonomous system
 \begin{equation}
 \dot{s} = - s  \left( -s + g \log{s} +c\right) 
 \end{equation}
 which can be solved implicitly as
 \begin{equation}\label{eq:ints}
\int \frac{d s}{s  \left( s - g \log{s} -c\right)}= \tau
\end{equation}
 Finally, the \textit{r} variable can also be conveniently expressed in terms of \textit{i}.
 For this purpose we solve the differential equation
 \[
\frac{dr}{di} = \frac{g}{s - g} = \frac{-1}{1 + W_{\pm} \left(- \frac{e^\frac{i-c}{g}}{g} \right) }
 \] 
 Therefore,
 \[
 r = c_1 -g \log {\left( -g W_{\pm}\left(-\frac{e^\frac{i -c}{g}}{g} \right) \right)   }
 \]
 by Prop. \ref{prop:indefint2}.
 On the other hand,
 \begin{multline*}
 g \log {\left( -g W\left(-\frac{e^\frac{i -c}{g}}{g} \right) \right)}   =
 	 g\, \left( \log{\left( \frac{{{ e}^{\frac{i-c}{g}}}}{g}\right) }-{W}\left( -\frac{{{ e}^{\frac{i-c}{g}}}}{g}\right) \right) =\\
  	- g W\left(-\frac{e^\frac{i -c}{g}}{g} \right) +i-g \log{g} -c = s +i -g \log{g} -c
 \end{multline*}
So that
\[
r = g W\left(-\frac{e^\frac{i -c}{g}}{g} \right) -i + c_1
\]
For the purposes of curve fitting we assume that $i(-\infty)=r(-\infty)=0$.
Therefore, 
\[
c_1 =  - g W_{-}\left(-\frac{e^{-\frac{c}{g}}}{g} \right)
\]

\subsection{Peak value parametrization}\label{sec:peak}

The upper terminal of integration can be determined by the requirement for the real-valuedness of $i$. 
This value of $i$ is denoted as $i_m$; that is
\[
  W_{\pm} \left(- \frac{e^\frac{i_m-c}{g}}{g} \right) =-1
\]
Therefore, 
\begin{equation}\label{eq:im}
i_m = c - g \log{g} -g
\end{equation}
The peak-value parametrization is supported by the following result.
\begin{proposition}\label{prop:imax}
	i(t) attains a global maximum  $i=i_m=  c - g \log{g} -g $.   
\end{proposition}
\begin{proof}
	We use a parametrization for which $i(0)=i_m$. Then
	\[
	\dot{i} (\tau)= - g i \left( W_{\pm} \left(-  {e^{\frac{i-i_m}{g}-1}} \right) +1 \right) 
	\]
	It follows that
	\[
	\dot{i}(0)=0, \quad i(0)=i_m
	\]
	so $i_m$ is an extremum. 
	In the most elementary way since W(z) should be real-valued then 
	\[
	-  e^{\frac{i-i_m}{g}-1}\geq  -1/e \Longrightarrow \frac{i-i_m}{g} \leq 0
	\]
	Hence, $i \leq i_m$.
\end{proof}

If we consider formally the phase space $  \left( z \times  y= - g z \left( W_{\pm} \left(-  {e^{\frac{z-i_m}{g}-1}} \right) +1 \right)   \right) $ the following argument allows for the correct branch identification. 
For $i \rightarrow -\infty$ $W_{-} \left(-  {e^{\frac{z-i_m}{g}-1}} \right) \rightarrow -\infty $ so $y <0$; 
while $W_{+} \left(-  {e^{\frac{z-i_m}{g}-1}} \right) \rightarrow 0^{+} $ so $y>0$.
Therefore, if we move the origin as $t(0)=i_m$ then conveniently 
\begin{flalign}\label{eq:invip}
-\int^{i}_{i_m}\frac{ dz}{z   \left( W_{+} \left(-   {e^{\frac{z-i_m}{g}-1}}  \right) +1 \right) } &= g  \tau, \quad \tau>0 \\
-\int^{i}_{i_m}\frac{ dz}{z   \left( W_{-} \left(-   {e^{\frac{z-i_m}{g}-1}}  \right) +1 \right) } &= g  \tau, \quad \tau \leq 0
\end{flalign}

Furthermore, the recovered population under this parametrization is
\begin{equation}\label{eq:rim}
r = g W_{\pm}\left(- e^{\frac{i-i_m}{g}-1} \right)  -g W_{-}\left(- e^{\frac{-i_m}{g}-1} \right)-i 
\end{equation}
under the same choice of origin. 

\subsection{Initial value parametrization}\label{sec:ivp}

As customarily accepted, the SIR model can be recast as an initial value problem.
The indeterminate constant $c$ can be eliminated using the initial condition 
 \[
i_0 = -s_0 + g \log{s_0} +c
\]
Therefore,
\begin{equation}
 i=i_0 +s_0  -s +g \log{s/s_0} = 1 - s +g \log{\frac{s}{1-i_0}}
\end{equation}
For this case, the following autonomous differential equation can be formulated:
\begin{equation}\label{eq:diffi}
\dot{i}= - g i \left( W_{\pm} \left( - \frac{1-i_0}{g} e^{\frac{i-1}{g}} \right)  +1\right) 
\end{equation}
This can be solved implicitly by separation of variables as
\begin{flalign}
-\int_{i_0}^{i} \frac{ d z }{z \left( W_{-} \left( - \frac{1-i_0}{g} e^{\frac{z-1}{g}} \right) +1 \right) } = g \tau, \quad \tau \leq t_m \\
-\int_{i_0}^{i} \frac{ d z }{z \left( W_{+} \left( - \frac{1-i_0}{g} e^{\frac{z-1}{g}} \right) +1 \right) } = g \tau, \quad \tau > t_m
\end{flalign}
It is noteworthy that the time to the peak of infections $t_m$ can be calculated as 
\[
t_m=\int\limits_{0}^{\log{g/s_0}} \frac{du}{ s_0 e^u  - g u-(s_0+i_0)  }
\]
The result follows by considering the autonomous system \ref{eq:dids} and fixing the upper terminal of integration $s=g$.
However, by Prop. \ref{prop:int3} this definite integral can be evaluated only numerically. 

 \section{Is the incidence function "new"?}\label{sec:louv}
 
 The incidence \textit{i}-function of the SIR model appears to be an interesting object of study on its won. 
 One may pose the question about the representation of this function by other, possibly, elementary functions.
 The answer to this question is in the negative as will be demonstrated below.
 In precise terms, this function is non-Liouvillian. 
 However, this does not mean that the function can not be well approximated. Fortunately, this is the case as the \textit{i}-function can be approximated for a sufficiently wide domain of parameter values by Newtonian iteration.  
 \begin{definition}\label{def:elementary} An elementary function is defined as a function built from a finite
 number of combinations and compositions of algebraic, exponential and logarithm  functions under algebraic operations (+,-,., /)
  \end{definition}
  Allowing for the underlying field to be complex numbers -- $\mathbb{C}$, trigonometric functions become elementary as well.
 \begin{definition}[Liouvillian function]
 	We say that $f(x)$ is a Liouvillian function if it lies in some Liouvillian extension of $\left(C(x), ^\prime \right) $
 	for some constant field \textit{C}. 
 \end{definition}
As a first point we establish the non-elementary character of the integral in eq. \ref{eq:invi}. The necessary introduction to the theory of differential fields is given in the Appendix \ref{sec:diffields}.
From the work of Liouville it is known that if a function
$ F (x)= f(x) e^{q(x)}$, where $f$, $q$ are elementary functions, has an elementary anti-derivative of the form \cite{Rosenlight1969}
\[
\int F(x)dx = \int f e^q dx = h e^q
\]
for some elementary function $h(x)$ \cite{Conard2005}.
Therefore, differentiating we obtain
\[
f e^q= h^{\prime} e^q + h q^{\prime} e^q 
\]
so that if $e^q\neq 0$
\[
h^{\prime} + h q^{\prime}  =f
\]
holds. The claim can be strengthened to demand that $h$ be algebraic for algebraic $f$ and $q$.
\begin{theorem}\label{th:invinonl}
	The integrals 
	\[
 	I_{\pm} (\xi)=	\int\frac{ d \xi}{ \xi   \left( W_{\pm} \left(-  \frac{e^\frac{\xi-c}{g}}{g} \right) +1 \right) }
	\]
	are not Liouvillian.
\end{theorem}
\begin{proof}
We use $i_m$ parametrization. Let $c= i_m+ g - g \log{g}$.
The proof proceeds by change of variables -- first $\xi= g \log{y}-y g+g+i_m $; followed by
$ z=\log{((g \log{y}+i_m+g)/g)}$.
\begin{multline*}
I=	\int\frac{ d \xi}{\xi   \left( W_{\pm} \left(-   {e^{\frac{\xi-i_m}{g}-1}}  \right) +1 \right) }
=\int \frac{y-1}{y \left( g \log(y)-g y +i_m+g \right) \left(  W\left(-y e^{-y} \right)+1\right)  } dy  \\
=- \int \frac{dy}{y (g \log(y)-g y +i_m+g )   }=\frac{1}{g}
\int {\left. \frac{{{ e}^{z+\frac{{i_m}}{g}+1}}}{{{e}^{{{ e}^{z}}}}-{{ e}^{z+\frac{{i_m}}{g}+1}}}dz\right.}
\end{multline*}
since $W_{\pm}\left(-y e^{-y} \right)=-y$.
The last integral has the form
\[
 \int \frac{A e^z}{e^{e^z}- A e^z} dz   
\]
which allows for the application of the Liouville theorem in the form of Corr. \ref{th:algint}.
We can identify
\[
\int f e^z dx = h e^z, \quad f(z)=\frac{A}{e^{e^z}- A e^z}, \quad A= e^{\frac{i_m}{g}+1}
\]
so that 
\[
\frac{A}{e^{e^z}- A e^z} = h^\prime (z) +h (z)
\]
for some unknown algebraic $h(z)$.
Since the left-hand side of the equation is transcendental in \textit{z} so is the right-hand side.
Therefore,  the integrand has no elementary antiderivative. 
\end{proof}
The proof establishes also the validity  of  two additional propositions:
\begin{proposition}\label{prop:int2}
	The integral 
	\[
	I= \int \frac{dy}{y \left( g \log(y)-g y + c \right)   } 
	\]
	is not Liouvillian.
\end{proposition}
\begin{proposition}\label{prop:int3}
	The integral 
	\[
	I= \int \frac{dy}{y+c -e^y   } 
	\]
	is not Liouvillian.
\end{proposition}
\begin{proof}
	By change of variables $y=e^x$
	\[
		I= \int \frac{dy}{y \left(  \log(y)- y +  c \right)   }  =
		- \int \frac{dx}{e^x-x-c}
	\]
\end{proof}

 \begin{theorem}\label{th:inonl}
 	The incidence function $i(t)$, defined by the differential eq. \ref{eq:isr}, is not Liouvillian.
 \end{theorem}
\begin{proof}
For the present case, let us assume that $i(0)=i_m$ so that \textit{i} attains the maximum by Prop \ref{prop:imax}.
Therefore,
\[
i^\prime = - g i \left(W\left(-e^{\frac{i-i_m}{g}-1} \right) +1 \right) 
\]
Without loss of generality let $g=1$, which amounts to scaling of the solution by the factor of $1/g$.

Suppose further that 
\[
i-i_m-1= \log{u} - u
\]
for some algebraic function $u$ (log-extension case).
Then
\[
-\frac{i^\prime}{i} = W\left(-e^{\log{u}-u} \right)+1 =-u+1 
\]
On the other hand,
\[
- \frac{i^\prime}{i}= 
\log{\left( \log{u} - u+i_m+1\right)}^\prime =
	-\frac{1-u}{u \left(\log{u} - u+i_m+1\right) } u^\prime
\]
so that 
\[
-\frac{1-u}{u \left(\log{u} - u+i_m+1\right) } = \frac{1-u}{u^\prime}
\]
However, if $u$ is algebraic so is $ u^\prime$ by Th. \ref{th:elemdiff}. 
Therefore, we have a contradiction, since the left-hand-side is transcendental.
Hence, $i$ is not part of a logarithmic  elementary extension. 

Suppose that $u$ is exponential, i. e. $u=e^f$ for some algebraic function $f$.
In this case,
\[
\frac{\left( {{ e}^{f}}-1\right) \, f^\prime }{-{{ e}^{f}}+f+i_m+1}=1-{{ e}^{f}} \rightarrow -f^\prime= -{{ e}^{f}}+f+i_m+1   
\]
However, 
\[
f^\prime +f+i_m+1   ={{ e}^{f}}
\]
Therefore, the right-hand-side is exponential and can not be algebraic as it is demanded by Corr \ref{th:algint}.
This is a contradiction, hence $f$ can not be algebraic, hence $u$ is not part of an exponential  elementary extension. 

Finally, suppose that  $i(t)$ is algebraic. 
Since the Wright function $\Omega (z)$ is transcendental \cite{Corless2002} it follows that $W\left(-e^{i-i_m-1} \right)= \Omega \left(i-i_m-1 +i \pi  \right)  $ can not be algebraic in \textit{i}. 
Therefore, $i^\prime/i$ and hence  \textit{i} must be transcendental.
Hence, the case of an algebraic $i(t)$ can not hold either.

In summary all three cases are rejected, therefore, \textit{i} is not Liouvillian. 
\end{proof}
\begin{remark}
	Prelle and Singer \cite{Prelle1983} prove in Corollary 3 that if the autonomous system $y^\prime = f(y)$ has an elementary first integral then 
	\[
	 g(y)=\int \frac{dy}{f(y)}
	\]
	is also elementary. 
	This presents a direct way of proving that the i-function is non elementary by virtue of Th. \ref{th:invinonl} but leaves open the question about the non-Liouvilian character of the i-function.
\end{remark}

\section{Series solutions}\label{sec:series}
We will give two series for the \textit{i}-function in view of the different re-parametrizations of the SIR model.

\subsection{Series for thr $i_m$ parametrization}\label{sec:taylorim}
 
The natural parametrization is fixing the peak at the origin.
The Taylor development can be computed as follows:
\begin{equation}
i (t)= {i_m}-\frac{{{{i_m}}^{2}} g}{2}{t}^{2}+\frac{{{{i_m}}^{3}} g\, }{6}{{t}^{3}}+\frac{\left( 4 {{{i_m}}^{3}} {{g}^{2}}-{{{i_m}}^{4}} g\right)  }{24}{{t}^{4}}-\frac{\left( 15 {{{i_m}}^{4}} {{g}^{2}}-{{{i_m}}^{5}} g\right)  }{120}{{t}^{5}}+ \ldots
\end{equation}
and for the logarithm
\begin{equation}
\log{i(t)}=\log{{i_m} }-\frac{{i_m} g}{2} {{t}^{2}}+\frac{{{{i_m}}^{2}} g }{6}{{t}^{3}}+\frac{\left( {{{i_m}}^{2}} {{g}^{2}}-{{{i_m}}^{3}} g\right)  }{24}{{t}^{4}}-\frac{\left( 5 {{{i_m}}^{3}} {{g}^{2}}-{{{i_m}}^{4}} g\right)  }{120}{{t}^{5}}+ \ldots
\end{equation}

\subsection{Series for the $i_0$-parametrization}\label{sec:taylori0}
The Taylor series starting from an initial value $i_0$ is
\begin{multline}
i(t) = {i_0}-{i_0} \left( {i_0}+g-1\right) t +\frac{g\, {i_0} \left( 4 {{{i_0}}^{2}}+5 g\, {i_0}-7 {i_0}-3 g+3\right) }{2 \left( {i_0}-1\right) } t^2 \\
-\frac{g\, {i_0} \left( 5 {{{i_0}}^{3}}+21 g\, {{{i_0}}^{2}}-9 {{{i_0}}^{2}}+18 {{g}^{2}}\, {i_0}-26 g\, {i_0}+4 {i_0}-10 {{g}^{2}}+7 g\right) }{6 \left( {i_0}-1\right)  }t^3 +\ldots
\end{multline}

The Taylor series for the logarithm starting from an  initial value $i_0$ is
\begin{multline}
\log{i(t)} =\log{\left( {i_0}\right) }-\left( {i_0}+g-1\right) t +\frac{g\, {i_0} \left( {i_0}+2 g-1\right) }{2 \left( {i_0}-1\right) } t^2 \\
 -\frac{g\, {i_0} \left( {i_0}+2 g-1\right)  \left( 2 {i_0}+g-1\right) }{6 \left( {i_0}-1\right) } t^3 + \ldots
\end{multline}
The series follow directly from successive differentiation of the differential equation \ref{eq:diffi}.

\section{Numerical approximation}\label{sec:numer}
 
The i-function can be efficiently approximated by the Newton's method.
The Newton iteration scheme is given as follows for the c-parametrization: 
 \[
 i_{n+1}=i_n-i_n\, \left( {W}_{\pm}\left( -\frac{{{ e}^{\frac{i_n-c}{g}
 }}}{g}\right) +1\right) \, \left( \int_{ g  \log{g} -g +c}^{i_n}{\left. \frac{d \xi}{\xi\, \left( {W}_{\pm}\left( -\frac{{{ e}^{\frac{\xi-c}{g}}}}{g}\right) +1\right) }\right.}+g t\right) 
 \]
 This is a conceptually simple representation. However, it has the disadvantage of using the Lambert function for the quadrature routine. 
 Another equivalent representation is
 \[i_{n+1}=i_n+g i_n\, \left( {W}_{\pm}\left( -\frac{{{ e}^{\frac{i_n-c}{g}}}}{g}\right) +1\right) \, \left( t-\int_{g}^{-g {W}_{\pm}\left( -\frac{{{ e}^{\frac{i_n-c}{g}}}}{g}\right) }{\left. \frac{dy}{y\, \left( g \log{y}-y+c\right) }\right.}\right) \]
(see Prop. \ref{prop:chvar}).
 This form has the advantage of requiring only 1 Lambert function evaluation per iteration.
 
 A point of attention here is the choice of the initial value for the iteration scheme. 
 Despite my best efforts, a rigorous analytical asymptotic valid on the entire real line and for all parameter values could not be found.  
 Numerical experiments gave acceptable results using the formula
 \[
 f(x) = b \, e^{1-x c - e^{-x c}}
 \]
 $g > 0, \ c = \sqrt{2 b g } $ or $ c = 2 \sqrt{ b g/e }$,
 and additionally 
 $c \leftarrow c /\sqrt{e}, \, x>0$
 for the initial value of $i_0 = f(t)$.

 \begin{figure}[h!]
 	
 	\begin{tabular}{ll}
 		A & B \\
 		\includegraphics[width=0.5\linewidth]{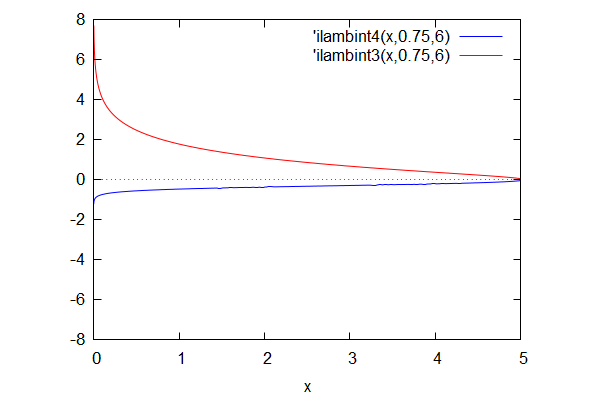} &
 		\includegraphics[width=0.5\linewidth]{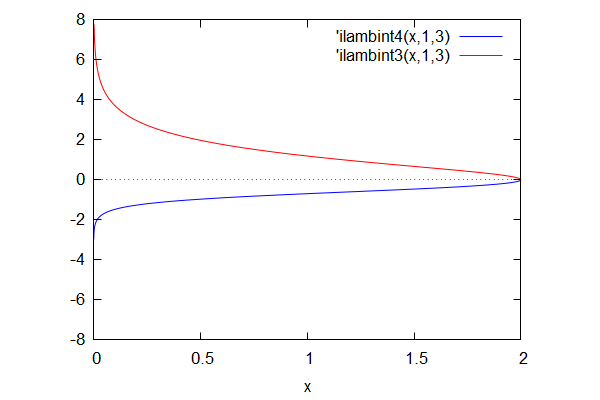}
 	\end{tabular}
 	
 	\caption{Plots of the integrals $I_{\pm} (x)$  }
 	\label{fig:ilambint0756}
 	A -- c-parametrization with parameters g= 0.75, c = 6.0; B -- c-parametrization with parameters g=1.0, c= 3.0. 
 	The negative branch is below 0.
 \end{figure}
 
 \begin{figure}[h!]
 	\centering
 	\includegraphics[width=0.8\linewidth]{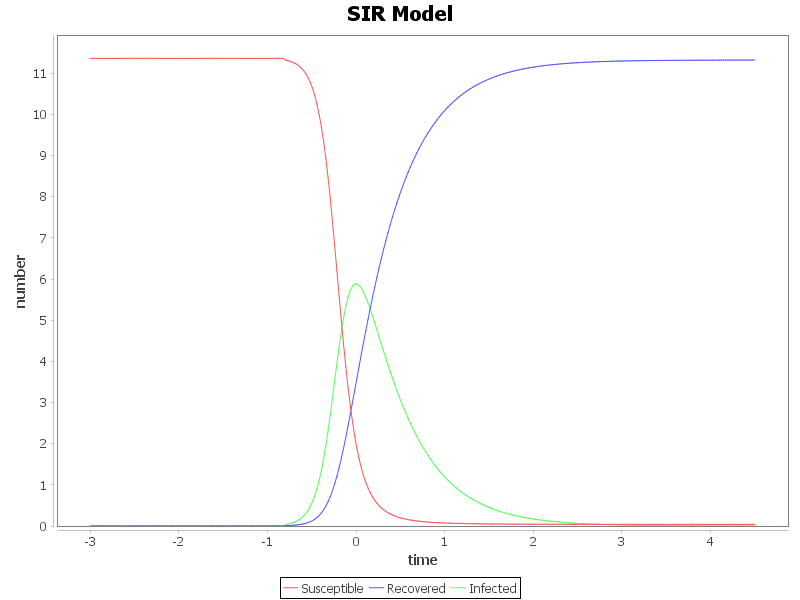}
 	\caption{The SIR model variables as functions of time}
 	\label{fig:sirchart3}
 	The instance is parametrized by $i_m=6.5$, $g=2.0$.
 \end{figure}

 \subsection{Plots}\label{sec:plots}
 
 Plots of the branches of the integral $I(x)$ (Fig. \ref{fig:ilambint0756}) were obtained by direct numerical integration using the QUADPACK \cite{Piessens1983} routines in the Computer Algebra System Maxima. 
 Plots of the SIR model (Fig. \ref{fig:sirchart3}) were obtained using a Java routine \cite{Prodanov2020} implementing the double exponential integration method \cite{Mori1985,Mori2001}.
 Both methods turn out to be suitable for the numerical integration problem. 
 
 \section{Datasets}\label{sec:data}
 The COVID datasets were downloaded from the European Centre for Disease Prevention and Control (ECDC) website: 
 \url{https://opendata.ecdc.europa.eu/covid19/casedistribution/csv}.
The downloadable data file is updated daily and contains the latest available public data on COVID-19. Each row/entry contains the number of new cases reported per day and per country.
The data collection policy is available from \url{https://www.ecdc.europa.eu/en/covid-19/data-collection}.  

 \section{Data analysis}\label{sec:app}
  
 \subsection{Processing}\label{sec:proc}
 The data were imported in the SQLite \url{https://www.sqlite.org} database and further filtered by country transferred to MATLAB for parametric fitting using native routines.
 Fitted parameters were stored in the same database. 
  
 \subsection{Parametric fitting}\label{sec:fit}
 The parametric fitting was conducted using the  \textsf{fminsearchbnd} routine, which allows for constrained optimization.
 To reduce the impact of the fluctuation in the weekly reporting of data the parametric fitting procedure was applied first to 7 day moving average of the time series.   
 The fitting algorithm is exemplified with datasets from Belgium, the Netwherlands, Italy, Germany and Bulgaria for the period Jan 2020 -- Jun 2020. 
 The fitting equation is given by 
 \[
 I_t \sim N I \left( t/10.0 - T | g, i_m \right) 
 \]
 where $I_t$ is the observed incidence or mortality, respectively, while $N, T, g$ and $i_m$ are estimated from the data.
 For numerical stability reasons the time during the fitting procedure is rescaled by a factor of 10. 
   
 \section{Results}\label{sec:results}

  The observed case fatality represented a parameter, which could be established with more confidence in the beginning of the pandemic due to the lack of testing and the non-specificity of the clinical signs of COVID-19. 
  Hence, it was the primary target of the parametric fitting.

  \subsection{Analysis of case fatality data}\label{sec:mortality}
  
  The data fitting procedure is illustrated with the case fatality data of Belgium and are presented in Fig. \ref{fig:belgium}.
  The $T$, $N$ and $i_m$ are estimated from the observed data. 
  For the $g$ parameter an initial estimate of 0.75 is used (i.e. $R_0=1.33$). 
  The intermediate parameters are initially estimated on the 3-day moving average data. 
  The final fit was performed on the raw data, using the intermediate parameters as initial values. 
  The cumulative mortality data were estimated from eq. \ref{eq:rim}.
  As can be appreciated from Fig. \ref{fig:belgium} fluctuations in reporting did not have a detrimental effect on the estimation procedure. 
  The results for Germany, Italy, Belgium and the Netherlands are presented in Table \ref{tab:mortality}.
  
   \begin{figure}[h!]	
   	\begin{tabular}{ll}
   	A & B \\ 
  	{\includegraphics[width=0.5\linewidth, clip, trim=3.5cm 9.5cm 3.5cm 9.5cm]{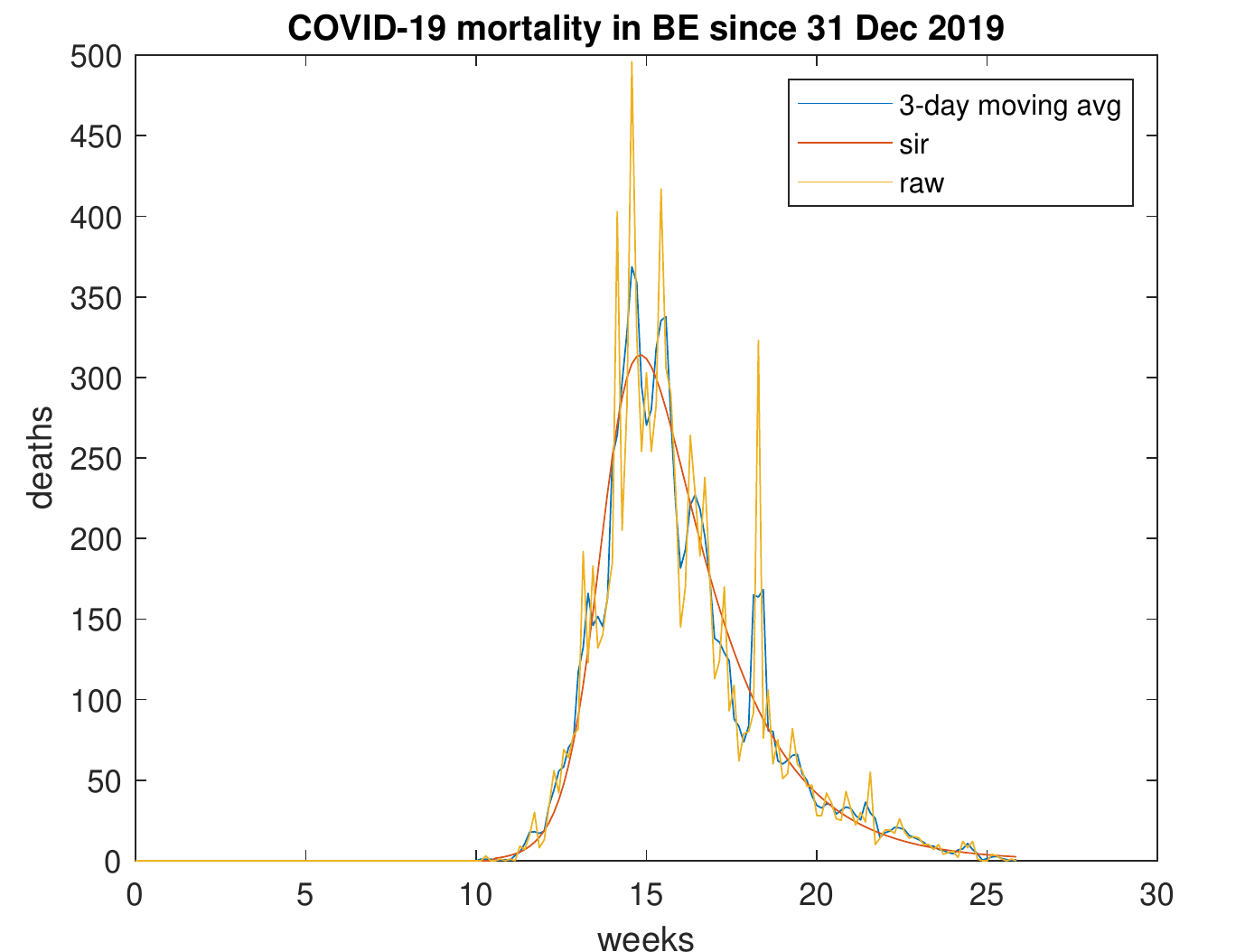}} &
  	{\includegraphics[width=0.5\linewidth, clip, trim=3.5cm 9.5cm 3.5cm 9.5cm]{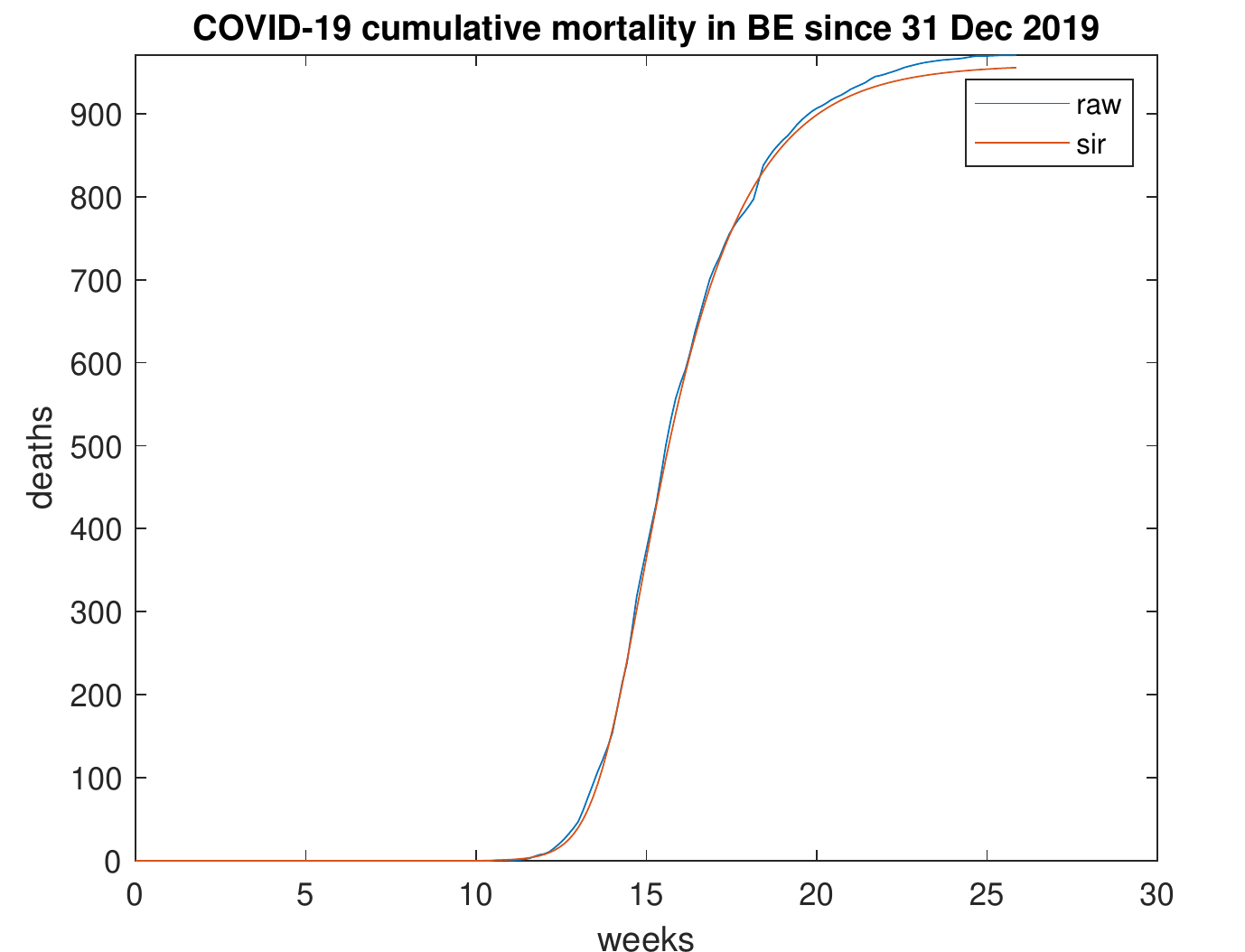}}
  \end{tabular}
  	\caption{Case fatality model for Belgium}
  	\label{fig:belgium}
  	A -- The parametric fit of the case fatality data;
  	B -- Cumulative deaths compared to the estimate from the r-variable.
  \end{figure}
    
  \begin{table}[h!]
  	\centering
  	\begin{tabular}{lllll}
  		\hline
  		Country  	&	g 	   & $R_0$	& T[weeks]  &  $i_m$ \\
  		\hline
  		Belgium	 	& 	0.7380 & 1.3549 & 14.8 		&	313.97 \\
		Netherlands &	0.5009 & 1.9962 & 14.1		&	156.72 \\
  		Germany 	& 	0.6109 & 1.6370 & 15.1 		&	226.51 \\
  		Italy		&   0.3734 & 2.6760 & 12.8		&	785.39 \\
  		\hline
  	\end{tabular}
  	\caption{Case fatality parameters}\label{tab:mortality}
  	T is given in weeks and refers to the time passed since 1\textsuperscript{st} Jan 2020.  
  \end{table}

  \subsection{Analysis of incidence data}\label{sec:morbidity}
  
  \begin{figure}[h!]
  	\begin{tabular}{ll}
  	A & B \\
  	\includegraphics[width=0.5\linewidth, clip, trim=3.5cm 9.0cm 3.5cm 9.5cm]{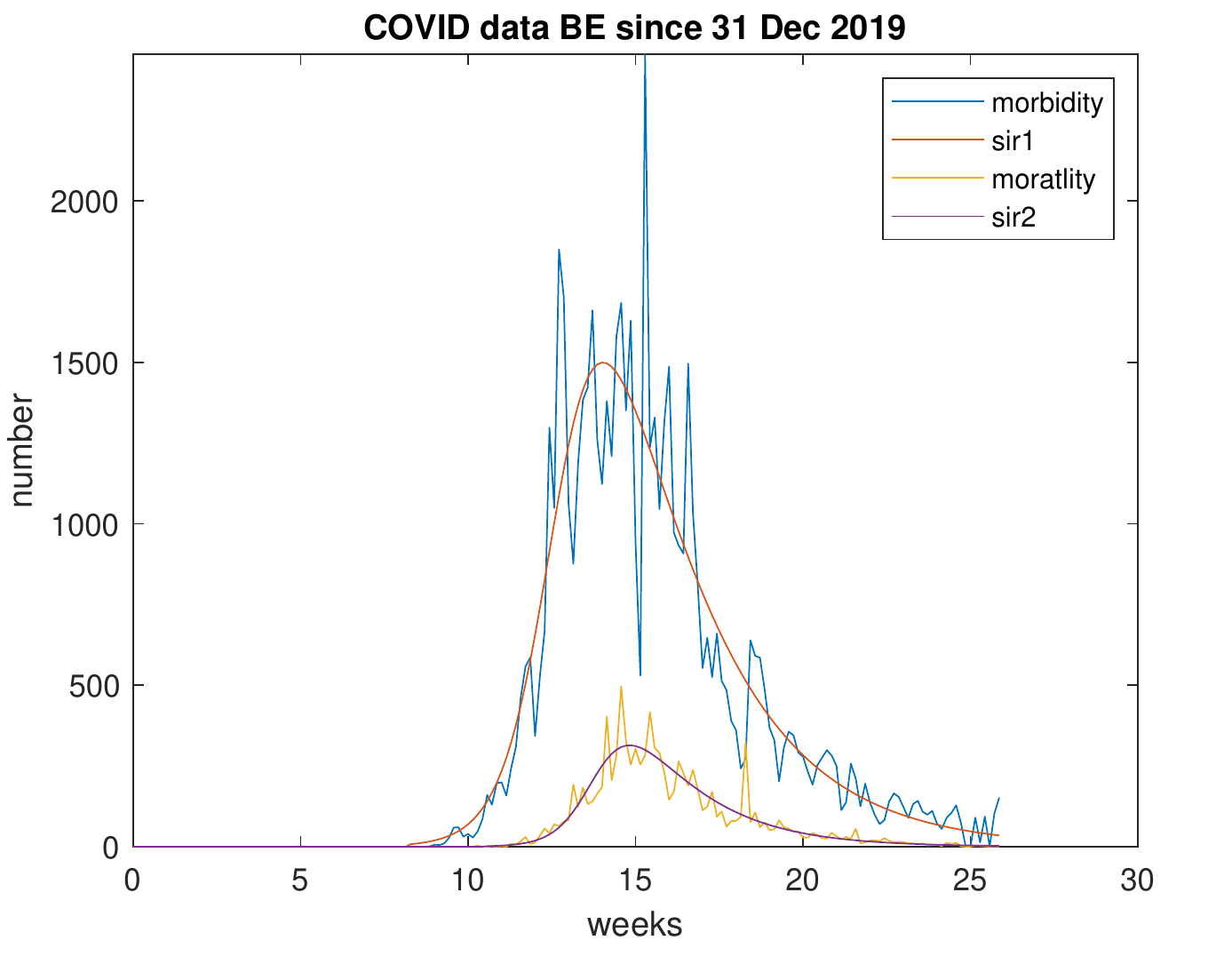} &
  	\includegraphics[width=0.5\linewidth,clip, trim=3.5cm 9.0cm 3.5cm 9.5cm]{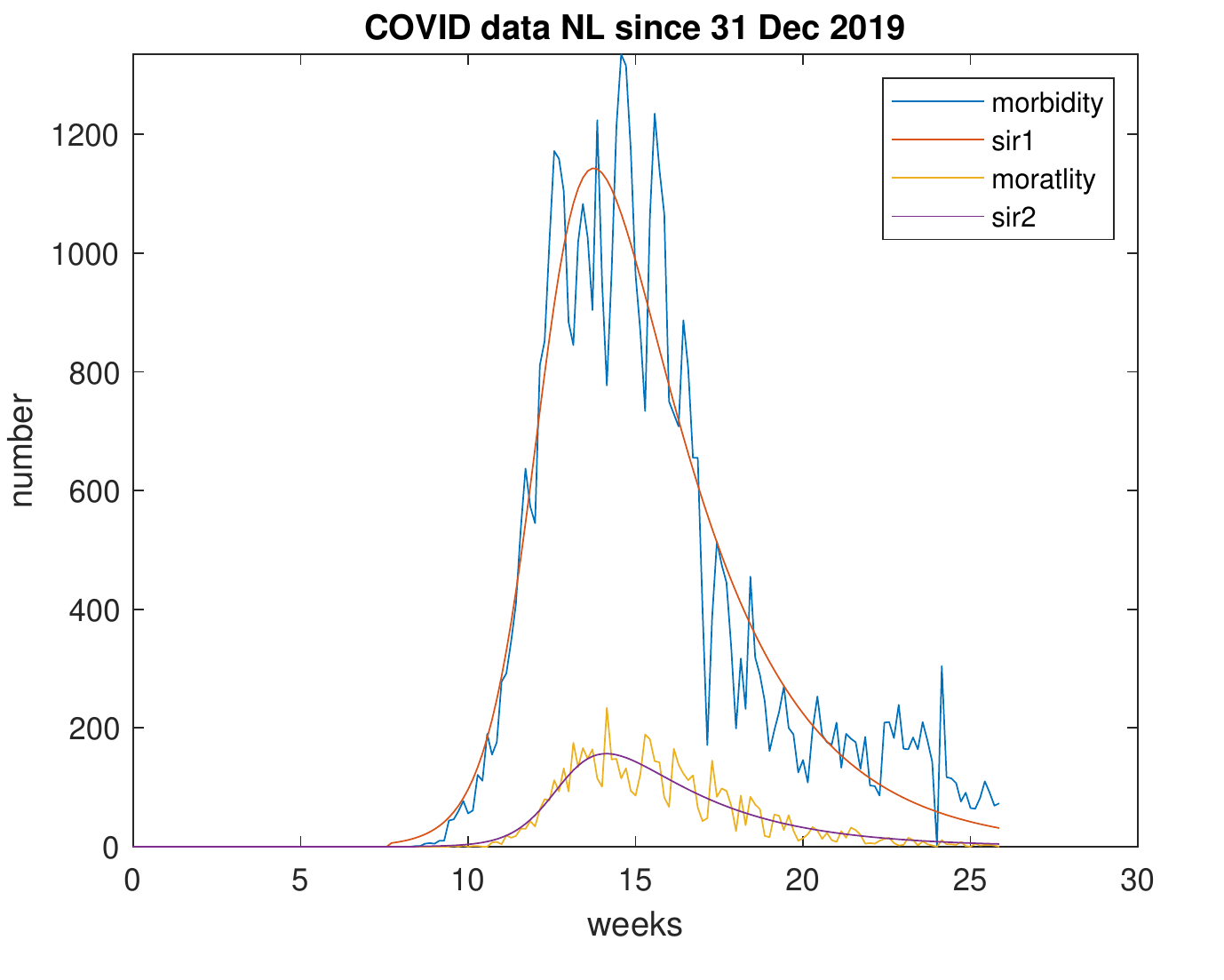} \\
  	C & D \\
  	\includegraphics[width=0.5\linewidth, clip, trim=3.5cm 9.0cm 3.5cm 9.5cm]{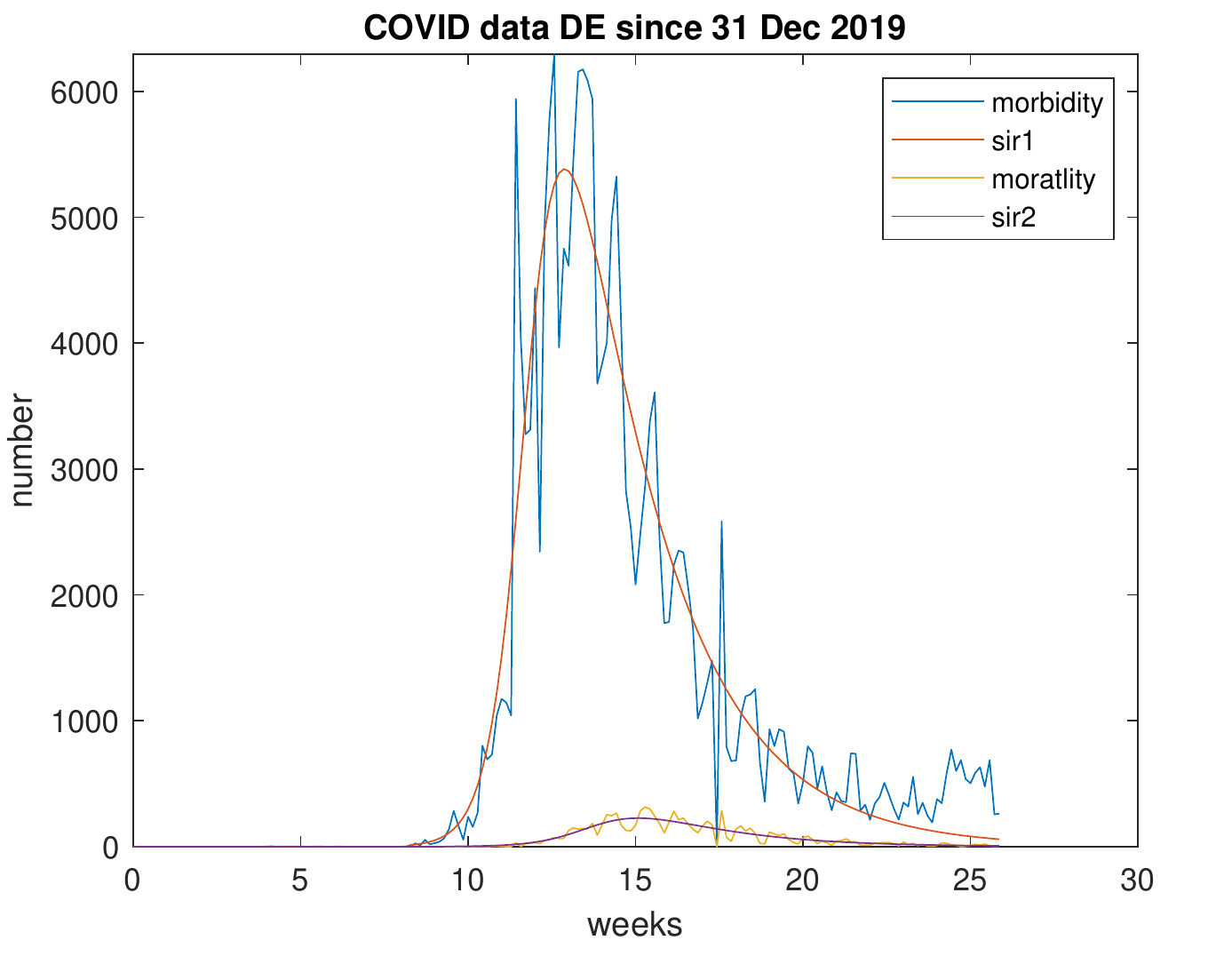} &
  	\includegraphics[width=0.5\linewidth,clip, trim=3.5cm 9.0cm 3.5cm 9.5cm]{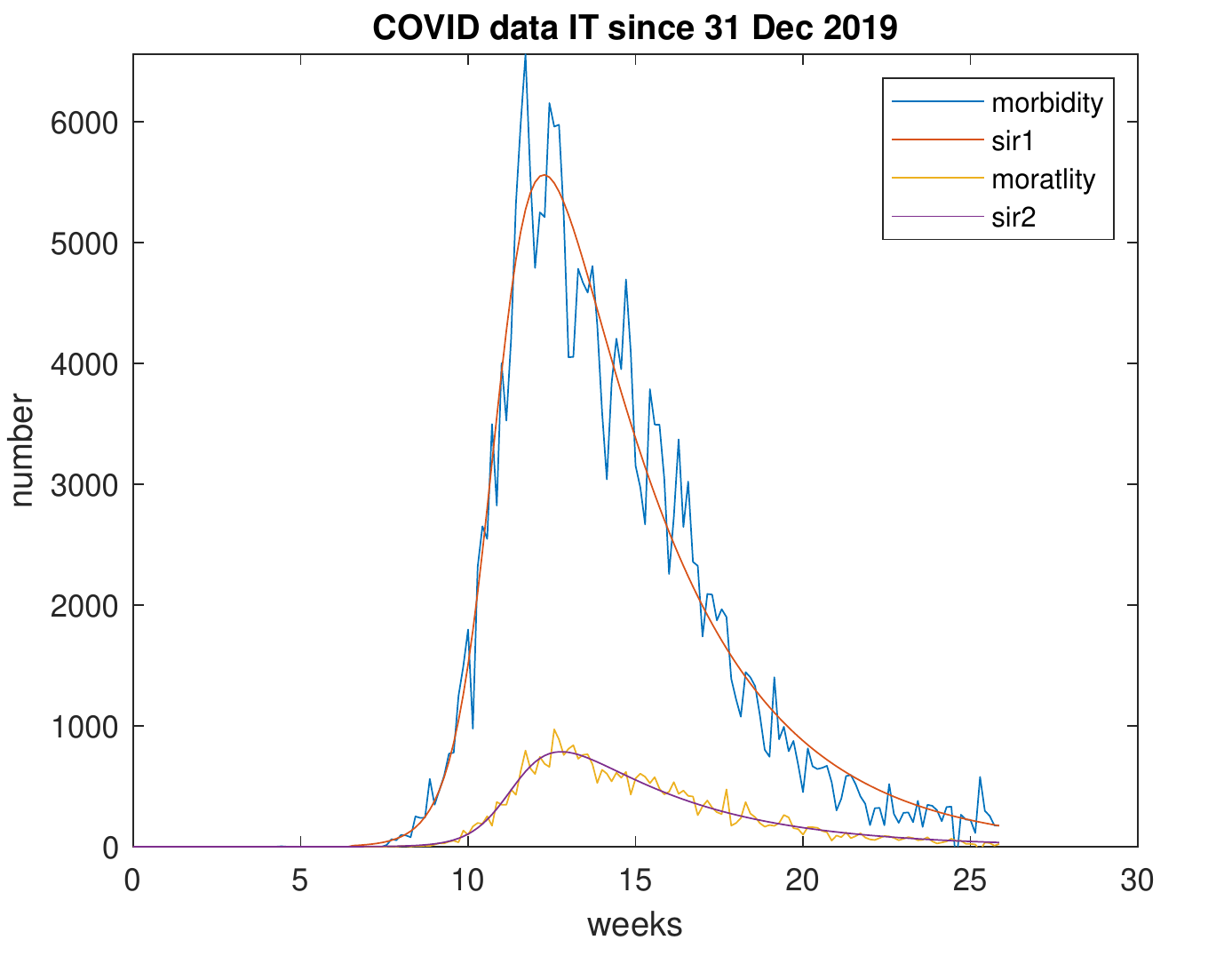} 	
  	\end{tabular}
 	\caption{Incidence and case fatality model fits for Germany and Italy}
	 \label{fig:germanyitaly}
	 
	 A -- combined data for Belgium by 29 Jun 2020;
	 B -- combined data for The Netherlands by 29 Jun 2020;
	 C -- combined data for Germany by 29 Jun 2020;
	 D -- combined data for Italy by 29 Jun 2020;
	 The raw data are smoothed with a 3-day moving average filter; 
	 mortality represents the case fatality, morbidity represents the incidence. 
  \end{figure}
  The raw data demonstrate weekly fluctuations most probably caused by the reporting irregularity due to holidays and fluctuations in the testing demand. 
  It is especially pronounced for the morbidity data of Germany for the presented period. 
  The incidence data were analysed in the same way using the fitted parameters for the mortality as initial values. 
  The fitting results demonstrate different lags of the peaks of incidence vs case fatality in the studied countries.
  For example for Germany it was 2.2 weeks, while for the Netherlands it was 0.3 weeks.  

\begin{table}[h!]
  \centering
\begin{tabular}{lllll}
	\hline
	Country		&	g	   & $R_0$  & T[weeks] 	&	 $i_m$ \\
	\hline
	Belgium	    &	0.5500 & 1.8183	&	14.0 	&	1499.59 \\
	Netherlands &	0.5149 & 1.9420 &	13.8	&	1143.25 \\
	Germany	    &	0.5483 & 1.8237	&	12.9 	&	5383.26 \\
	Italy	    &	0.4006 & 2.4961	&	12.3 	&	5560.30 \\
	\hline
\end{tabular}
  \caption{Incidence parameters}\label{tab:incidence}
   T is given in weeks and refers to the time passed since 1\textsuperscript{st} Jan 2020;  $R_0=1/g$; $i_m$ corresponds to the peak of the case fatality.  
\end{table}

  \subsection{Tracking of multiple outbreaks}\label{sec:multiple}
  
  Changes of containment policies are meant to result in changes in the epidemic outbreak dynamics. 
  This can be followed by the SIR model as demonstrated in the Bulgarian incidence dataset, where resuming of public sports events in the end of June correlates with the 3\textsuperscript{rd} and ongoing increase of incidence (Fig. \ref{fig:bulgaria3}).   
  The process is difficult to automate because of the fluctuations in the data. 
  Nevertheless, it was possible to accurately track the past outbreaks.  
 
  \begin{figure}[h!]
  	\centering
  	{\includegraphics[width=0.7\linewidth,clip, trim=3.5cm 9.0cm 3.5cm 9.5cm]{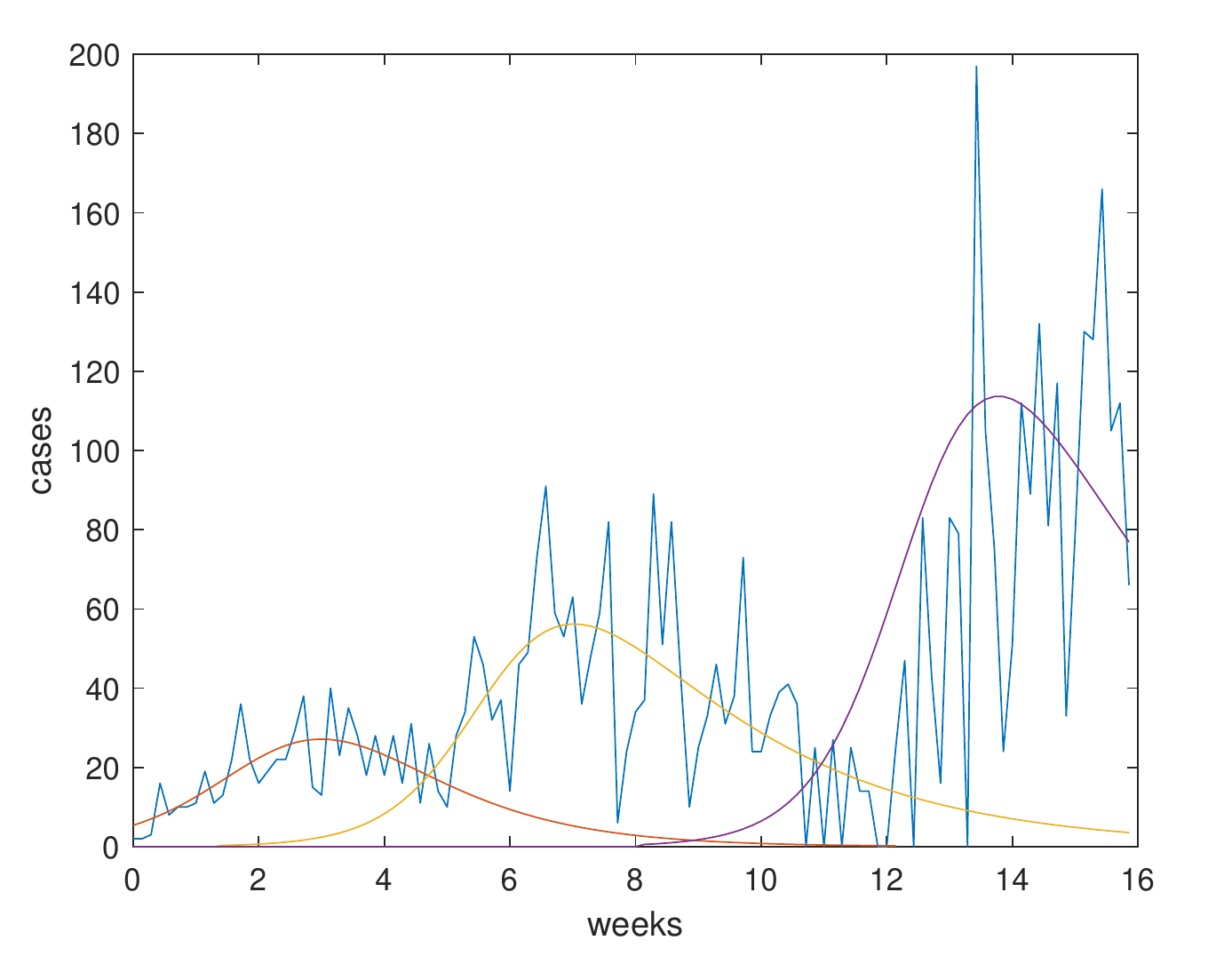}}
  	\caption{Modelling of consecutive outbreaks in Bulgaria by 29 Jun 2020}
  	\label{fig:bulgaria3}
  	Raw data are compared to the fitted outbreaks. 
  	The origin corresponds to 8 March 2020 when the first COVID19 case was reported.  
  \end{figure}
  
  \section{Discussion}\label{sec:disc}

	The SIR model was formulated to model epidemic outbreaks \cite{Kermack1927}.The model has only two independent variables and two parameters, which allow for  their estimation from data. There is a renewed interest in this model in view of the coronavirus disease 2019 (COVID-19) pandemics \cite{Fanelli2020, Postnikov2020, Cooper2020, Record2020}.
	For decades the model evaded the efforts of the community to derive explicit solution. 
	A formal analytical solution of the SIR model has been found only recently and formulated in the traditional setting of an initial value problem \cite{Harko2014}.
	Yet, another perspective can be also of merit -- the model can be treated as a manifold problem, which can be parametrized by any point on the flow. 
	This allows for efficient curve-fitting approach as demonstrated by the presented results.
	The approach is exemplified with data from the European Centre for Disease Prevention and Control (ECDC) for several European countries in the period Jan 2020 -- Jun 2020. 
	The presented results can be discussed along three main directions.
	
	\subsection{Analytical}
	  The present manuscript establishes novel analytical results about the SIR model.
	  From mathematical perspective, these are mostly the non-elementarity results for most of the presented integrals. 
	  Another interesting point is the proof of the non-Liouvillian character of the incidence \textit{i}-function.
	  An alternative proof could be also given, in principle, based on the work of Prelle and Singer \cite{Prelle1983}.

  \subsection{Numerical}
  Presented results demonstrate the robustness of the fitting procedure with regard to the fluctuations in the raw data. 
  On the other hand, more efforts are necessary in establishing a robust asymptotic of the incidence \textit{i}-function.
  This is a clear direction for future research. 
 
  On the second place, the analytical formula involving first exponentiation and then computation of the Lambert W function has disadvantages for large arguments due to float under or overflows \cite{Corless2002}.
  Therefore, another special function can be used in principle, notably the Wright $\Omega$ function. 
  On the other hand, optimized routines for its calculation are not readily available in MATLAB.  
  
  \subsection{Epidemiological implications of the results}
   
  Recorded data necessarily suffer from diverse biases. 
  For example, the marked difference in the availability of tests in the early stages of the pandemics in different countries. 
  Another one is the unsteady reporting resulting in fluctuations of the numbers. 
  This severely limits the usefulness of the model formulation as an initial value problem, which is the standard mathematical assumption.  
  This can result is drastic overestimation of the infection peak (see for example the predictions for UK in \cite{Ferguson2020}).   
  
  Presented results indicate that there is a universality in the time evolution of COVID-19 and the same epidemic model, notably SIR, can be applied to countries having large differences in populations sizes and densities. 
  More interestingly, the model seems to fit well also the mortality data, which can be interpreted in the sense that the vulnerable population forms a distinct subpopulation from all susceptible individuals (e.g. elderly people).
  Presented data lend support to a simple modification of the SIR model:  notably -- the SIRD model with an independent population of dead (D) persons. 
  This corresponds to the recent findings of other authors \cite{Fanelli2020, Cooper2020}.
  
  A key finding of the present report is that simple models can be very useful in studying the epidemic outbreaks. 
  This can be eventually extended to predicting the effects of different containment measures or the lack thereof \cite{Record2020}. 

 \appendix 
 \section{Special functions}\label{sec:specfunct}
 
\subsection{The Lambert W function}\label{sec:lambert}

The Lambert W function can be defined implicitly by the equation
\[
W(z)e^{W(z)}=z, \quad z \in \fclass{C}{}
\]
We observe that by Lemma \ref{th:comp} $W(z)$ is transcendental.
Furthermore, the Lambert function obeys the differential equation for $x \neq -1$
\[
W(x)^\prime=\frac{e^{-W(x)}}{1+ W(x)}
\]
W is a multivalued function. It has many applications in pure and applied mathematics.
Properties of the W function are given in \cite{Corless1996}.
Useful identities 
\begin{flalign}\label{eq:lambident}
e^{-W(z)} &=\frac{W(z)}{z} \\
e^{n W(z)} &=\left( \frac{z}{W(z)}\right)^n \\
\log{W(z)} &= \log{z} -W(z) \\
W \left(\frac{n z^n}{W(z)^{n-1}} \right) &= n W(z), \quad n>0, z>0
\end{flalign}
The W function is non-elementary and in particular it is non-Liouvillian \cite{Bronstein2008}. 
Its indefinite integral is:
\[
\int W(x)dx=x \operatorname{W}(x)+\frac{x}{\operatorname{W}(x)}-x
\]

\begin{proposition}\label{prop:indefint2}
	\[
	\int \frac{dy}{1+ W\left(-\frac{e^\frac{y-c}{g}}{g} \right) }   =g \log {\left( -g W\left(-\frac{e^\frac{y-c}{g}}{g} \right) \right)}   + const
	\]
\end{proposition}
\begin{proof}
We differentiate 
\[
g \left(\log W\left(-\frac{e^\frac{y-c}{g}}{g} \right)\right)^\prime=
 -\frac{e^{\frac{y-c}{g}- W\left(-\frac{e^\frac{y-c}{g}}{g} \right)}}{g  W\left(-\frac{e^\frac{y-c}{g}}{g}\right)  \left( 1+  W\left(-\frac{e^\frac{y-c}{g}}{g} \right) \right) }=
\frac{1}{1 + W \left(- g e^\frac{i-c}{g} \right) }
\]
\end{proof}
\begin{proposition}\label{prop:chvar}
	\[
	\int_{ g  \log{g} -g +c}^{i}{\left. \frac{d \xi}{\xi\, \left( {W}_{\pm}\left( -\frac{{{ e}^{\frac{\xi-c}{g}}}}{g}\right) +1\right) }\right.} = 
	\int_{g}^{-g {W}_{\pm}\left( -\frac{{{ e}^{\frac{i-c}{g}}}}{g}\right) }{\left. \frac{dy}{y\, \left( g \log{y}-y+c\right) }\right.}
	\]
\end{proposition}
\begin{proof}
	We use the change of variables $ \xi -c= g \log{y}-y$  and then simplification by the defining identity of the Lambert W function.
	\begin{multline*}
		\int_{ g  \log{g} -g +c}^{i}{\left. \frac{d \xi}{\xi\, \left( {W}\left( -\frac{{{ e}^{\frac{\xi-c}{g}}}}{g}\right) +1\right) }\right.} = 
		\int_{A}^{B} \frac{\frac{g}{y}-1}{\left( g \log{y}-y+c\right) \, \left( {W}\left( -\frac{{{ e}^{\frac{g \log{y}-y+c}{g}-\frac{c}{g}}}}{g}\right) +1\right)  } dy= \\
		\int_{A}^{B} 
		\frac{g-y}{-\frac{y\, \left( g y \log{y}-{{y}^{2}}+ c y\right) }{g}+g y \log{y}-{{y}^{2}}+c y} dy =
		g\int_{A}^{B} \frac{ dy}{y\, \left( g \log{y}-y+c\right) }
	\end{multline*}
	where 
	\[
	g \log{A} - A +c= g \log{g} -g +c, \quad g \log{B} - B+c=i
	\]
	Therefore, $A=g$ and $B= -g {W}\left( -\frac{{{ e}^{\frac{i-c}{g}}}}{g}\right) $.
\end{proof}
\subsection{The Wright $\Omega$ function}\label{sec:omega}

The Wright $\Omega$ function is related to the Lambert W function \cite{Corless1996}
\[
\Omega(z)=W_{K(z)} \left( e^z\right) , \quad z \in \fclass{C}{}
\]
where $ K(z) = \lceil \left( Im(z) - \pi\right)  / 2 \pi \rceil$ is the unwinding number of $z$.
Moreover, 
\[
\Omega(z) + \log {\Omega(z)} = z, \quad z \neq t  \pm i \pi, t \leq-1,
\]
for the principal branch of the logarithm.
It is a transcendental function.
It obeys the differential equation
\[
\Omega(x)^\prime=\frac{\Omega(x)}{1+ \Omega(x)}
\]
The  $\Omega$ function is non-Liovillian \cite{Bronstein2008}.
Its indefinite integral is:
\[
\int \Omega(x) dx =\frac{{{\Omega(x)}^{2}}}{2}+\Omega(x)+ const
\]

\section{Differential fields}\label{sec:diffields}

\begin{definition}
	Denote by $\fclass{C}{}\left( x, c_i, \theta_i\right) $ the complex-valued ring, generated by the finite set of rational functions $\{\theta_i \}_i^n$ and constants $\{c_i \}_i^n$. 
\end{definition}
\begin{definition}
	An element $\theta$ is called \underline{algebraic} if $P(x, \theta) = 0$ for some polynomial 
	\[
	P(x, t) = t^m + a_{m-1}t^{m-1} + \ldots + a_0,
	\]
	where $a_i$ can be also rational functions of $x$,
	or else it is called \underline{transcendental}.
\end{definition}
\begin{lemma}[Composition lemma]\label{th:comp}
	Denote by $a$ and $t$ the algebraic or transcendental elementary functions, respectively.
	The following compositions hold
	\[
	a \circ a = a, \quad t \circ a = t,  \quad a \circ t = t
	\]
\end{lemma}
\begin{proof}
 The case $a \circ a$ when $a(x)$ is a polynomial is trivial.
 Suppose that a and b are both algebraic:
  \[
  P(x, a) =0, \quad Q(x, b)=0
  \]
  Without loss of generality suppose that $a_i$ are polynomial.
  Formally, $b=\bar{f}_k(x)$ for any branch $k$ with $\bar{f}$ algebraic since it is a root of a polynomial, where the bar denotes the inverse function in order to avoid confusion with exponentiation.
  Therefore,
  \[
  a \circ b =  b^m + a_{m-1}b^{m-1} + \ldots + a_0 = \bar{f}_k^m (x) +a_{m-1} \bar{f}_k^{m-1} (x)  + \ldots + a_0 
  \]
  is algebraic since it is computed by a finite sequence of algebraic operations.
  
 Suppose that $t=exp(x)$.
 Then  $exp(a)$ is not algebraic, hence it is transcendental.
 \[
 P(x, e^x) = e^{x m} + a_{m-1}e^{x m- x} + \ldots + a_0
 \]
 is exponential.
 
 Suppose that $t=log(x)$.
 Then  $log(a)$ is not algebraic, hence it is transcendental.
 \[
 P(x,\log(x)) =\log^m{x} + a_{m-1}\log^{m-1} {x}+ \ldots + a_0
 \]
 is not algebraic, hence it is transcendental.
\end{proof}

In what follows is assumed that the differential field is of characteristic zero and has an algebraically closed field of constants. An element y of a differential field is said to be an \textit{exponential} of an element \textit{A} if $y^\prime = A y$, 
an exponential of an integral of an element \textit{A} if $y^\prime =A y$;
logarithm of an element \textit{A} if $y^\prime = A^\prime/A$, and an integral of an element \textit{A} if $y^\prime= A$.

The next definition is due to \cite{Bronstein2008}.
\begin{definition}
	Let $(k, ^\prime \equiv d/dx)$ be a differential field of characteristic 0.
	A differential extension $(K, ^\prime \equiv d/dx)$ of \textit{k} is called Liouvillian over k if there are $\theta_1, \ldots, \theta_n \in K$, such that
	$K=C(x, \theta_1, \ldots, \theta_n )$ and for all i, at least one of the following 
	\begin{enumerate}
		\item $\theta_i$ is algebraic over $k(\theta_1, \ldots, \theta_{n-1} )$
		\item $\theta_i^\prime \in k(\theta_1, \ldots, \theta_{n-1} )$
		\item $\theta_i^\prime/\theta_i \in k(\theta_1, \ldots, \theta_{n-1} )$
	\end{enumerate}
  holds.
 The constant subfield $C(K)$ of K is defined to be the set of c in K, such that $c^\prime = 0$.
\end{definition}

The next theorem is due to \cite{Conard2005}.
\begin{theorem}\label{th:elemdiff}
	If K is an elementary field, then it is closed under differentiation.
\end{theorem}
An elementary integrability theorem due to Conrad \cite{Conard2005}.
\begin{theorem}[Rational Liouville criterion]\label{th:conrad}
	For $f, g \in \fclass{C}{}(x)$ with $f$ and $g$ non-constant the function $f(x)e{^g(x)}$ can be integrated in elementary terms if and only if there exists
	a rational function $h \in \fclass{C}{}(x)$ such that $h^\prime + g^\prime h  = f$.
\end{theorem}
The last result can be extended to algebraic functions as follows.
\begin{corollary}[Algebraic Liouville criterion]\label{th:algint}
	For $f(x), g(x)$ algebraic and non-constant, the function $f(x)e^{g(x)}$ can be integrated in elementary terms if and only if there exists an algebraic function $h(x)$, for which  $h^\prime + g^\prime h  = f$.
\end{corollary}
\begin{proof}
Suppose that $f$ and $g$ are arbitrary elementary algebraic functions.
Denote the primitive of $f$ as $f \div F$.
The integral can be integrated by parts
\[
 I = \int f(x)e{^{g(x)}} dx = \int  {e^g(x)} dF=F (x) e^{g(x)} - \int  F(x)  \left(e^{g(x)}\right) ^\prime dx
\]
Therefore,
\[
\int \left( f(x)+ F(x)  g^\prime(x) \right) e^{g(x)} dx= F (x) e^{g(x)}
\]
We observe that $ g^\prime(x)$ is elementary by Th. \ref{th:elemdiff}.
The L.H.S has the form $f e^g$ and since $f(x)+ F(x)  g^\prime(x) $ is elementary we can identify
\[ 
h \equiv F, \quad f_1  \equiv f + F g^\prime = h^\prime + h g^\prime 
\]
so that $ \left( h^\prime + h g^\prime\right) e^g= \left( h e^g\right)^\prime  $
and the claim follows.
\end{proof}

 \bibliographystyle{plain}
 \bibliography{infbib}
\end{document}